\documentclass[11pt]{article}
\usepackage{hyperref, graphicx, enumerate, authblk}
\usepackage{fullpage, amsmath, amsthm, bm, float}

\newcommand{\ud}{\,\mathrm{d}} 

\newtheorem{proposition}{Proposition}

\title{\bf{A Constrained Conditional Likelihood Approach for Estimating the Means of
Selected Populations}}
\author[1]{Claudio Fuentes}
\author[2]{Vik Gopal}
\affil[1]{Department of Statistics, Oregon State University}
\affil[2]{Department of Statistics and Applied Probability, National University of Singapore}
\date{}

\begin{document}
\maketitle

\begin{abstract}

Given $p$ independent normal populations, we consider the problem of estimating
the mean of those populations, that based on the observed data, give the
strongest signals. We explicitly condition on the ranking of the sample means,
and consider a constrained conditional maximum likelihood (CCMLE) approach,
avoiding the use of any priors and of any sparsity requirement between the
population means.
Our results show that if the observed means are too close together, we should
in fact use the grand mean to estimate the mean of the population with the
larger sample mean. If they are separated by more than a certain threshold, we
should shrink the observed means towards each other. As intuition suggests,
it is only if the observed means are far apart that we should conclude that the
magnitude of separation and consequent ranking are not due to chance. 
Unlike other methods, our approach does not need to pre-specify the number of selected populations and the proposed
CCMLE is able to perform simultaneous inference.
Our method, which
is conceptually straightforward, can be easily adapted to incorporate other
selection criteria.
\end{abstract}



\section{Introduction}
Consider a scenario where $p \ge 2$ independent normal populations are
available, and from each one of them, we obtain a sample of size $n$. In this
context, practitioners are sometimes interested in estimating the true means of
those populations that yielded the $k$ largest sample means in the experiment. 

A naive solution to the problem is to estimate the selected population means
with the corresponding sample means. Such an approach, however, is known to be
problematic.  \cite{putter1968estimating} showed that the resulting estimator is
biased for the case $k=1$. This bias is particularly evident when all $p$
populations are identically distributed. 
In terms of optimality, \cite{stein1964minimax, sackrowitz1986evaluating} both
showed that the estimator is minimax only when $p=2$.

In order to improve on the naive estimator, several alternatives have been
proposed in the literature, including \cite{dahiya1974estimation},
\cite{cohen1982estimating}, and \cite{cohen1986estimating}. These papers
propose estimators that perform better in terms of the Mean Squared Error (MSE).
\cite{venter1988estimation} considered a bias correction approach for the
problem, obtaining estimators that perform well in terms of frequentist risk.
Following up on his own idea, \cite{venter1991estimation} introduced 
$\omega$-estimators which are essentially a weighted average of the order
statistics. \cite{cohen1989two} considered a two-stage procedure, assuming we
can obtain a second sample from the selected population, to produce unbiased
estimators of the selected means. Despite these results, performance theorems
are scarce, with the exception of  \cite{bolotskikh2016post} and 
\cite{hwang1993empirical}. The latter proposes an
empirical Bayes estimator and shows that it performs better in terms of the
Bayes risk with respect to any normal prior. The former paper 
focuses on admissibility and obtains a generalized Bayes estimator under a
harmonic prior.

More recently, \cite{reid2014post} make an implicit assumption of sparsity,
that many effect sizes (absolute population means) $\theta_i=0$. By adapting 
the theory developed in \cite{lee2014exact}, the authors in \cite{reid2014post} 
perform post-selection inference with the Lasso. \cite{simon2013estimating}
approached the problem by estimating the first and second order bias of the
naive estimator. Their results, which are similar in performance to the
empirical Bayes approach in \cite{efron2011tweedie}, extend to the non-Gaussian
setting.

In this paper we propose a new estimator, that is based entirely on the
likelihood function after incorporating the selection process.  We motivate and
define this new estimator in section \ref{sec:CCMLE_def}, where we also provide
a neat result for $p=2$ that yields some insight into how the estimator works.
In section \ref{sec:CCMLE_comp}, we discuss the main computational hurdle in
computing this estimator, and provide some direction on overcoming it. Following
that section, we summarize the results of a simulation study that highlights the
benefits and flaws of this new estimator. We conclude the paper with a brief
discussion of our main results.

\section{Conditional Likelihood Estimation}
\label{sec:CCMLE_def}

\subsection{Defining the Conditional Likelihood}

For simplicity, suppose that we obtain a single observation $X_i \sim
N(\mu_i, \sigma^2)$ from each population, and that the common variance
$\sigma^2=1$. Then, for $\bm{\mu} = (\mu_1, \ldots, \mu_p)^T$, the unconditional
likelihood is given by $L_0(\bm{\mu}) = \prod_{i=i}^p \phi(x_i - \mu_i)$,
where $\phi$ denotes the density function of the standard normal distribution.
However, once we observe $X_1=x_1,\;X_2=x_2,\;\ldots,\;X_p=x_p$, we rank the
observations in order to identify the populations corresponding to the largest
$x_i$'s and estimate the respective means $\mu_i$.
Hence, the conditional likelihood of interest is 
\begin{equation}
\label{eq:cond_lik}
L(\bm{\mu}) = 
\frac{\prod_{i=i}^p \phi(x_i - \mu_i)}{P_{\bm{\mu}}(X_1 > X_2 \ldots > X_p)},
\end{equation}
where the notation $P_{\bm{\mu}}(\cdot)$ explicitly states that the probability
under consideration depends on $\bm{\mu}$. Note that, for equation
(\ref{eq:cond_lik}), we do not have to worry about the labels attached to the
groups. In other words, there is no loss of generality in assuming that the
ordering of the means is $X_1 > X_2 > \ldots > X_p$. It follows from equation
(\ref{eq:cond_lik}) that the log of the conditional likelihood is


%
\begin{equation}
\label{eq:cond_loglik}
l(\bm{\mu}) =  C - \frac{1}{2} \sum_{i=1}^p (x_i - \mu_i)^2 - 
\log{P_{\bm{\mu}}(X_1 > X_2 \ldots > X_p) }
\end{equation}

 \subsection{Constrained Conditional MLE}
\label{subsec:unb_cl}

For the case $p=2$, note that $\bm{\mu} = (\mu_1, \mu_2)$ and 
\begin{equation}
\label{eq:p_2_prob}
P_{\bm{\mu}}(X_1 > X_2) = P_{\bm{\mu}}(X_1 - X_2 > 0) = 1 - \Phi \left( \frac{\mu_2 - \mu_1}{\sqrt{2}} \right)
\end{equation}
For a fixed $\mu_2$, $P_{\bm{\mu}}(X_1 > X_2)$ approaches 0 as $\mu_1\downarrow-\infty$.
This means that the $-\log P_{\bm{\mu}}(X_1 > X_2)$ term in equation (\ref{eq:cond_loglik})
increases to $\infty$, causing the conditional likelihood to be unbounded.  In
other words, there is no global maximum for the expression in 
(\ref{eq:cond_loglik}) for $-\infty < \mu_1, \mu_2 < \infty$. This phenomenon is
of course, not specific to $p=2$.

The occurrence of an unbounded likelihood is not without precedent in the
statistics literature. Possibly the most widely studied models in which this
occurs is a normal mixture model with unequal variances
\cite{hathaway1985constrained}. One of the solutions in that model was to constrain
the parameter space to where there are local modes, and that is what we shall
attempt to do here as well.

Since we expect that the populations with the largest sample means would be
the ones with the largest population means, we now aim to maximise the
conditional likelihood in equation (\ref{eq:cond_loglik}) subject to 
$\bm{\mu} \in \Theta$, where $\Theta = \{ \bm{\mu} : \mu_1 \ge \mu_2 \ge 
\mu_3 \ldots \ge \mu_p \}$.
We thus define the Constrained Conditional Maximum Likelihood Estimator (CCMLE) to be 
\begin{equation}
\label{eq:ccmle_def}
\hat{\bm{\mu}} = \underset{\bm{\mu} \in \Theta}{\mbox{arg max}}\; 
l(\bm{\mu})
\end{equation}

This is clearly a case of constrained statistical inference. Much of the theory
regarding this approach can be found in  \cite{silvapulle2011constrained}.

%
\subsection{A Closed Form Result}
\label{subsec:CCMLE_p2}

%
%
%
\begin{proposition}
\label{thm:p_2}
Consider the CCMLE when $p=2$ and the variance is known. Thus, we have
realisations $x_1$ and $x_2$ from $N(\mu_i,\; \sigma^2)$ for $i=1,2$, and have
observed that $x_1 > x_2$. If $x_1 - x_2 > \frac{2 \sigma}{\sqrt{\pi}}$, then
the unique CCMLE is an interior point of $\Theta = \{(\mu_1, \mu_2): \mu_1 \ge
\mu_2 \}$. If $x_1 - x_2 \le \frac{2 \sigma}{\sqrt{\pi}}$, the CCMLE is
$\hat{\bm{\mu}} = (\bar{x}, \bar{x})$, where $\bar{x} = (x_1 + x_2)/2$.
\end{proposition}


\begin{proof}
Let us first define $g(\cdot)$ to denote the inverse Mills ratio:
\begin{equation}
\label{eq:mills_1}
g(x) = \frac{\phi(x)}{1 - \Phi(x)}
\end{equation}
We shall use certain properties of $g$ in this proof -- in particular, note that 
$g$ is convex and monotone increasing \cite[see][]{baricz2008mills}.

Recall that the log-likelihood function is given by
\[
l(\bm{\mu}) = -\log 2\pi - \log \sigma^2 - \frac{1}{2\sigma^2} 
\sum_{i=1}^2 (x_i - \mu_i)^2 - 
\log{\left( 1 - \Phi \left(\frac{\mu_2 - \mu_1}{\sqrt{2 \sigma^2}} \right) \right)}
\]
%

Setting the partial derivatives to zero, we find that, if a solution exists, it
must satisfy
\begin{equation}\label{eq:sol}
g \left( \sqrt{\frac{2}{\sigma^2}} (\bar{x} - \mu_1) \right) = 
\sqrt{\frac{2}{\sigma^2}} ( x_1 - \mu_1 )
\end{equation}
Observe that equation (\ref{eq:sol}) depends only on $\mu_1$. Let us define the
function on the left to be $h_1(\cdot)$ and the function on the right to be
$h_2(\cdot)$. Thus, a stationary point exists if $h_1$ and $h_2$ intersect. Note
that we are finding the root of a transcendental equation; no closed form
solution exists. A plot of these two functions can be seen in Figure
\ref{fig:noroot_1}. Several quantities on the plot can be derived. For instance,
from equation (\ref{eq:mills_1}), we can compute that $h_1(\bar{x}) = g(0) =
\sqrt{\frac{2}{\pi}}$ and $ h'_1(\bar{x}) = -g'(0) \sqrt{2 /\sigma^2} \approx
-0.900/ \sqrt{\sigma^2}$.  Moreover, the height $H$ is $(x_1 - x_2)/\sqrt{2
\sigma^2}$.
 
\begin{figure}[H]
\centering
\includegraphics[width=0.6\textwidth]{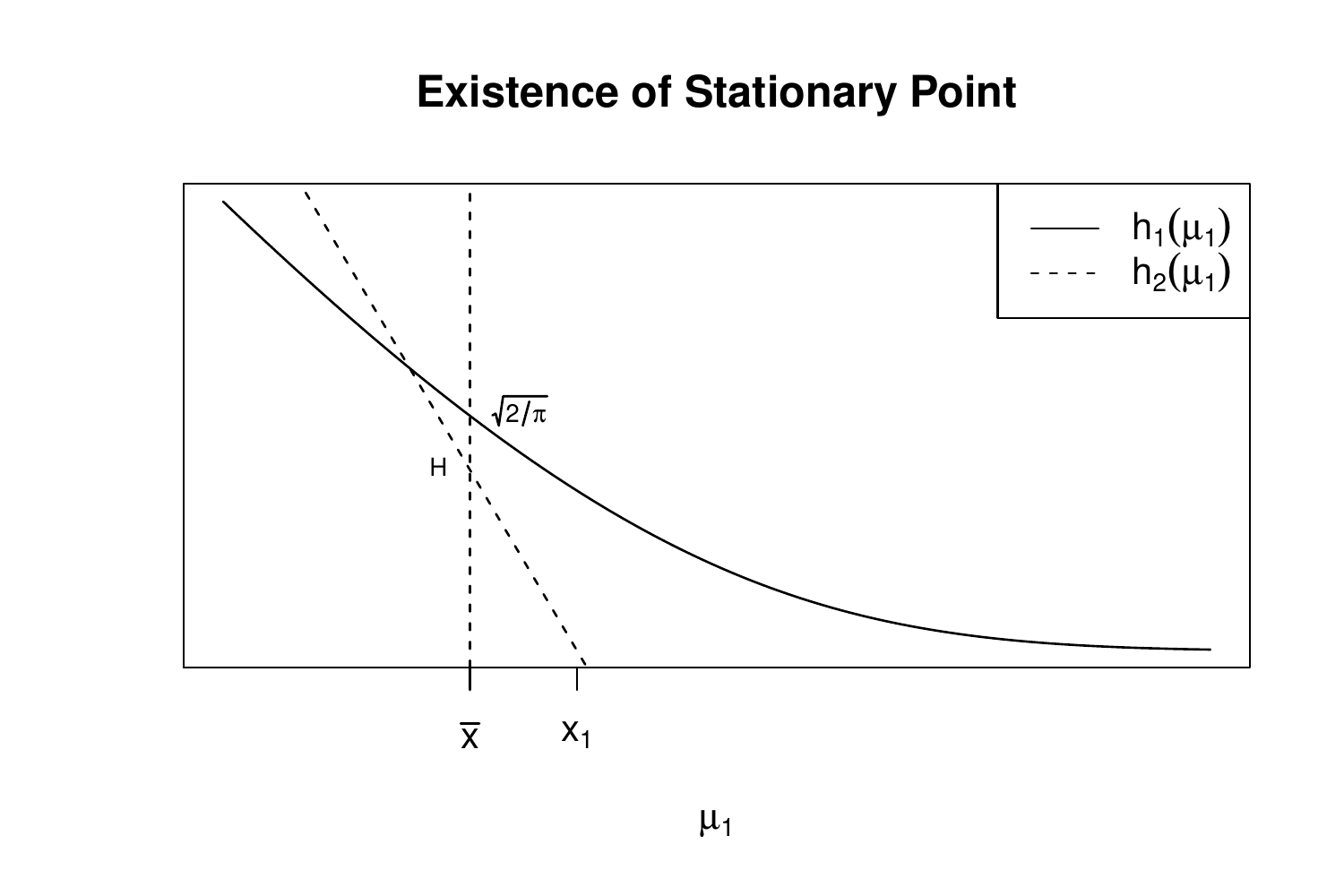}
\caption{Conditions for existence/absence of a stationary point.}
\label{fig:noroot_1}
\end{figure}

The figure shows that an intersection exists in $\{ \mu_1: \mu_1
\ge \bar{x} \}$ if and only if $H > \sqrt{2/\pi}$. In other words, 
if and only if $x_1 - x_2 > 2 \sigma/\sqrt{\pi}$.

Suppose that there is more than 1 root in $\{ \mu_1: \mu_1 \ge \bar{x} \}$.
Then, at some point, it must be that $h'_1(\mu_1)< - \sqrt{2 / \sigma^2}$.
However $h'_1(0) \approx -0.900 / \sqrt{\sigma^2}$, and $h_1$ is convex.
Hence $h'_1(\mu_1) > -0.900 / \sqrt{\sigma^2} $ for all $\mu_1 \ge
\bar{x}$. It follows from this contradiction that the stationary point (if it
exists) is unique.

Our final task is to show that, when $x_1 - x_2 \le 2 \sigma/\sqrt{\pi}$,  the
CCMLE is on the boundary of $\Theta$ at $(\bar{x}, \bar{x})$. Within the
constrained parameter space $\Theta$, we have $\mu_2 \le \mu_1$ and therefore 
\begin{equation*}
1 - \Phi \left( \frac{\mu_2 - \mu_1}{\sqrt{2 \sigma^2}} \right) \ge 0.5
\end{equation*}

Hence, for all $\bm{\mu} \in \Theta$, 
\[
l(\bm{\mu}) \le C - \frac{1}{2 \sigma^2} \sum_{i=1}^2 (x_i - \mu_i)^2 + \log 2
\]
It is then easy to see that as $\mu_1$ and/or $\mu_2$ go to $+\infty$ and/or
$-\infty$, it holds that $l(\bm{\mu}) \rightarrow -\infty$.

When $x_1 - x_2 < 2 \sigma/\sqrt{\pi}$, we already know that there are no stationary
points in $\Theta$. This information, coupled with the understanding that
$l(\bm{\mu})$ decreases without bound as $\bm{\mu}$ moves away from the
boundary, leads to the conclusion that the CCMLE must be on the boundary $\{
\bm{\mu}: \mu_1 = \mu_2 \}$.

Now consider any point $\bm{\mu}$ such that $\mu_1 = \mu_2$. If we can show that
the directional derivative in the direction of $(\bar{x}, \bar{x})$ is always
positive, we are done, because it implies that we can always increase the
log-likelihood by taking a suitable step in the direction of $(\bar{x},
\bar{x})$. To simplify notation, let  
\[
\gamma = \frac{1}{\sqrt{2 \sigma^2}} g \left( \frac{\mu_2 - \mu_1} {\sqrt{2
\sigma^2}} \right)
\]
Thus we can write the gradient of $l$ as $ \nabla l = 
\left( \frac{1}{\sigma^2}(x_1 - \mu_1) - \gamma,\; \frac{1}{\sigma^2}
(x_2 - \mu_2) + \gamma \right)^T$. The direction we need to consider is 
$ \bm{u} = K(\bar{x} - \mu_1,\; \bar{x} - \mu_2)^T$,
where $K$ is a positive normalising constant, that makes $|\bm{u}| = 1$. 

For any point in $\Theta$, the directional derivative is
\begin{equation}
\label{eq:dd_1}
\nabla l^T \bm{u} \propto ((1/\sigma^2)(x_1 - \mu_1) - \gamma)(\bar{x} - \mu_1) + 
((1/\sigma^2)(x_2 - \mu_2) + \gamma)(\bar{x} - \mu_2) 
\end{equation}
Now for any point on the edge of $\Theta$, we let $\mu = \mu_1 = \mu_2$ and
simplify expression (\ref{eq:dd_1}) to show that 
\begin{equation*}
\nabla l^T \bm{u} \propto \frac{2}{\sigma^2}(\bar{x} - \mu)^2 \ge 0
\end{equation*}
\end{proof}

\subsection{Sample Computations}
\label{subsec:CCMLE_examples}

It is not straightforward to generalize the proof in Proposition \ref{thm:p_2} to
higher dimensions. However, the same intuitive results and numerical approaches apply. 

In Table \ref{tab:sample_ccmle} below, we present some sample CCMLE
computations for various combinations of observed values when $p=4$. In
performing the computations, we assume that $\sigma^2$ is known, and is equal to 1.
The purpose is to underline how, whenever a subset of observed sample means are
close together, the CCMLE procedure will shrink them towards each other. This is
apparent in row 1 of the table. Unlike the $p=2$ case, however, the critical
distance at which they are collapsed onto one another is no longer
$2/\sqrt{\pi}$. This latter phenomenon can be observed in row 2, where the
separation is less than $2/\sqrt{\pi}$, but the estimated values are not exactly
equal.

\begin{table}[H]
\centering
{\small
\begin{tabular}{|l|cccc|cccc|}
\hline
& \multicolumn{4}{c|}{\textbf{Observed Values}} &
\multicolumn{4}{c|}{\textbf{Estimates}} \\
& $\bm{x_1}$ & $\bm{x_2}$ & $\bm{x_3}$ & $\bm{x_4}$ & $\bm{\hat{\mu}_1}$ &
$\bm{\hat{\mu}_2}$ & $\bm{\hat{\mu}_3}$ & $\bm{\hat{\mu}_4}$  \\
\hline
Config 1 & 10.0 & 9.5 & 9.0 & 0.0 & 9.50 & 9.50 & 9.50 & 0.00  \\ 
Config 2 & 10.0 & 9.0 & 8.0 & 0.0 & 9.35 & 9.00 & 8.65 & 0.00  \\ 
Config 3 & 10.0 & 9.0 & 1.0 & 0.0 & 9.50 & 9.50 & 0.50 & 0.50  \\ 
Config 4 & 10.0 & 2.0 & 1.0 & 0.0 & 10.00 & 1.35 & 1.00 & 0.65  \\ 
\hline
\end{tabular}
}
\caption{Sample CCMLE computations for $p=4$}
\label{tab:sample_ccmle}
\end{table}

%
%
%
%

\section{Computation of the CCMLE}
\label{sec:CCMLE_comp}

\subsection{Calculation of Probabilities}
As part of numerically maximising the conditional likelihood in equation
(\ref{eq:cond_loglik}), we need to repeatedly compute $P_{\bm{\mu}}(X_1 > X_2 \ldots
> X_p)$ for different $\bm{\mu}$ vectors.  In this section, we outline how it is
possible to compute this integral for small values of $p$ using a trick of
conditioning, followed by a nested application of the \texttt{integrate}
function in R.

For the case $p=2$, it is clear from equation (\ref{eq:p_2_prob}) that the
probability of interest can be computed using the usual approximations to the
standard normal distribution function.

For the case $p=3$, we can condition on the random variable in the middle to
yield a univariate integral.
\begin{eqnarray*}
P_{\bm{\mu}}(X_1 > X_2 > X_3) &=& \int P_{\bm{\mu}}(X_1 > X_2 > X_3 | X_2 = x_2) f(x_2) \ud x_2 \\
&=& \int P_{\bm{\mu}}(X_1 > x_2) P(x_2 > X_3) f(x_2) \ud x_2
\end{eqnarray*}

The same approach enabled us to compute the probabilities up to $p=7$ without
any further optimisation in R. For higher dimensions, we advocate computing the
integral using a lower level language such as C, and switching to a sparse grid
method \cite{heiss2008likelihood} instead of persisting with cubature
techniques.

%

\subsection{Obtaining A Good Starting Point}
\label{subsec:CCMLE_pg2}

%

We now focus on obtaining a good starting point for the numerical optimisation
in the case that $p > 2$. A good starting point ensures that we can reduce the
number of times that we evaluate the high-dimensional probability and its
derivatives. Let us first denote $f(\bm{\mu}) = 
\log{P_{\bm{\mu}}(X_1 > X_2 \ldots > X_p) }$. Taking a first order Taylor
approximation of $f$ about the observed $\bm{x}$, we can approximate the
conditional likelihood in (\ref{eq:cond_loglik}) with 
\begin{equation}
\label{eq:cond_loglik_taylor}
l(\bm{\mu}) \approx  C - \frac{1}{2}  (\bm{x} - \bm{\mu})^T (\bm{x} - \bm{\mu}) - 
f(\bm{x}) - (\bm{\mu} - \bm{x})^T \nabla f(\bm{x})
\end{equation}

Taking derivative with respect to $\bm{\mu}$ and setting the above equation to
0, we can get an approximation to the CCMLE. It works out to be 
\begin{equation}
\label{eq:taylor_est}
\hat{\bm{\mu}}_0 \approx  \bm{x} - \nabla f(\bm{x})
\end{equation}

The solution does not mean that a stationary point always exists - remember that
we are solving an approximation to the conditional log-likelihood. For the same
reason, it is also possible that the solution in equation (\ref{eq:taylor_est})
does not fall within the constrained space $\Theta$. In such cases, we shall use
the orthogonal projection onto $\Theta$ as the starting point. Note that the
projection has to be performed numerically - fortunately however, $\Theta$ is
convex, and hence we can rely on prior methods from convex optimization theory
in order to perform this step. The R package \cite{ravi2009BB} provides a good
implementation for solving this projection problem, using a quadratic
programming technique.

\section{A Comparison By Simulation}

In this section, we shall assess the performance of the CCMLE via its Mean
Squared Error, and the bootstrap confidence intervals constructed using it. We
take the opportunity to highlight that the errors and intervals have to respect
the selection procedure.

For instance, suppose that populations A, B and C have true means 3, 2 and 1
respectively. If the corresponding sample means are 2.1, 2.2 and 1.8, then the
population selected as the maximum would be population B. The error in
estimating the mean of the selected population using the sample mean would be 
\[
	2 - 2.2 = -0.2
\]
This is the methodology employed in Section \ref{subsec:mse}. The error should
not be computed as $3 - 2.2 = +0.8$, since the true mean of the selected
population is in fact 2.

Similarly, when bootstrapping the strata in Section \ref{subsec:boot}, the
population selected to have the maximum will not always be population B. It
could be population A or even population C depending on the bootstrap sample
drawn.  Thus it is not accurate to describe it as a confidence interval for
population A (which has the largest mean). It is a confidence interval for
the population \emph{selected} to have the maximum mean.

\subsection{Mean Squared Error Comparison}
\label{subsec:mse}
In this subsection, we conduct a simulation study to understand the MSE of the
CCMLE, as compared to the MSE of the ordinary MLE. We consider only the cases
when $p=2$ and $p=3$ as they are sufficiently informative. 

For the $p=2$ case, we fix $\mu_2 = 0$, and vary $\mu_1$ from 0 to 5. For each
configuration, we generate 1000 bivariate $N(\bm{\mu}, \bm{I})$ random vectors,
and then estimate the mean of the population with the larger sample mean.

The MSE estimate for each configuration has been plotted and smoothed in Figure
\ref{fig:plot_p2}. Notice that the CCMLE performs very well when $\mu_1 - \mu_2
\le 1.5$. Beyond that, it performs approximately 10\% worse than the
unconditional MLE until the difference in means becomes quite large. From such
a point onwards, it will be the case that the sample means will be far enough
apart to warrant no shrinkage at all. It seems reasonable to guess that when
populations are close together, the CCMLE will be very beneficial; however, when
the population means are in fact far apart, it is not the right estimator to
use. We shall witness this in the $p=3$ case as well.

\begin{figure}[H]
\centering
\includegraphics[width=0.7\textwidth]{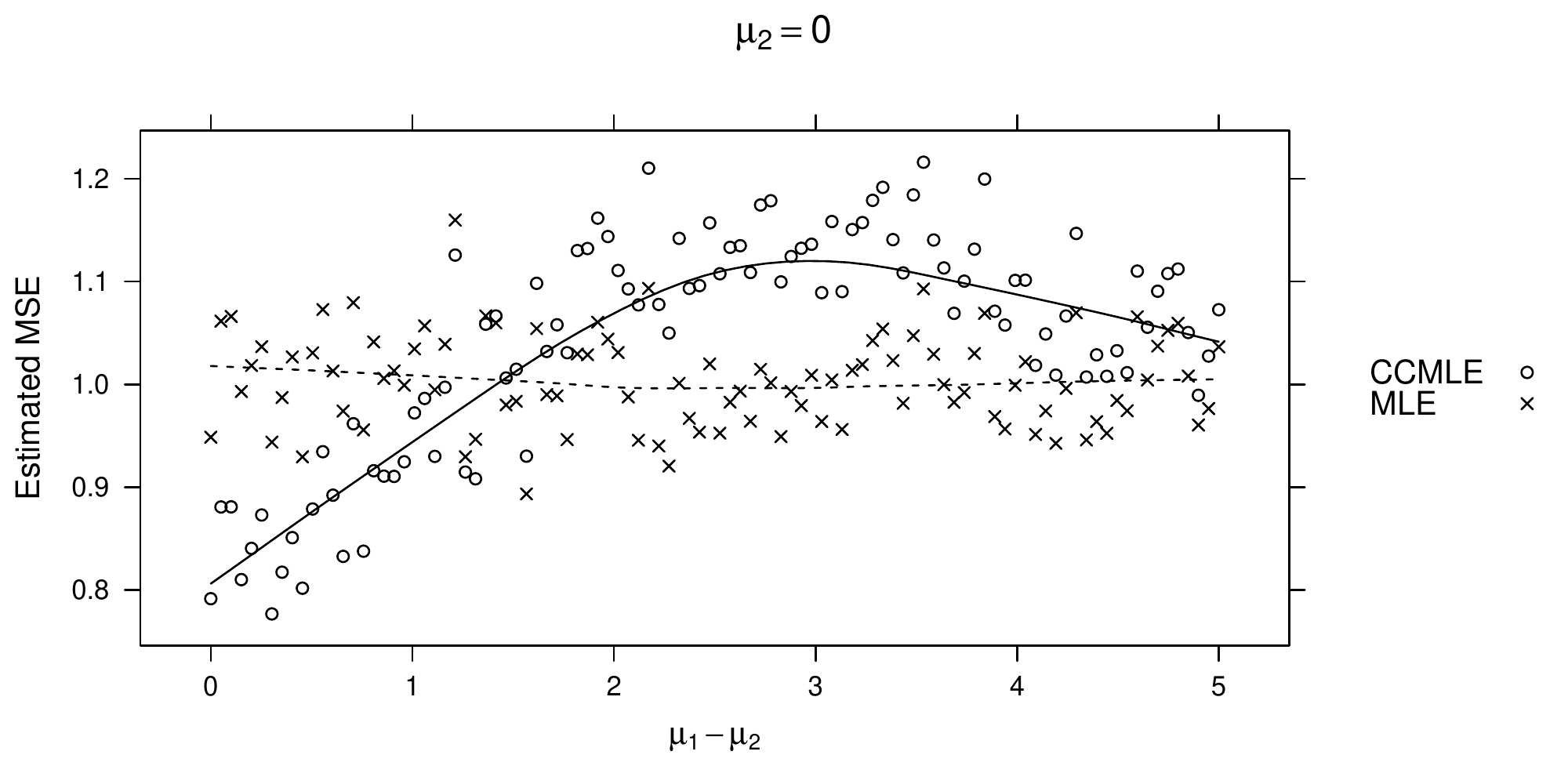}
\caption{Estimated MSE for $p=2$ simulation experiment.}
\label{fig:plot_p2}
\end{figure}

Now let us turn to the situation when $p=3$. In our experiment, we considered 3
possible values for $\mu_3:0, 2$ and 4. For each $\mu_3$, we varied $\mu_2$
and $\mu_1$ between $\mu_3$ and 5. The ouput for these experiments is shown in 
Figure \ref{fig:plot_p3_1}. 

\begin{figure}[ht]
\centering
\includegraphics[width=0.6\textwidth]{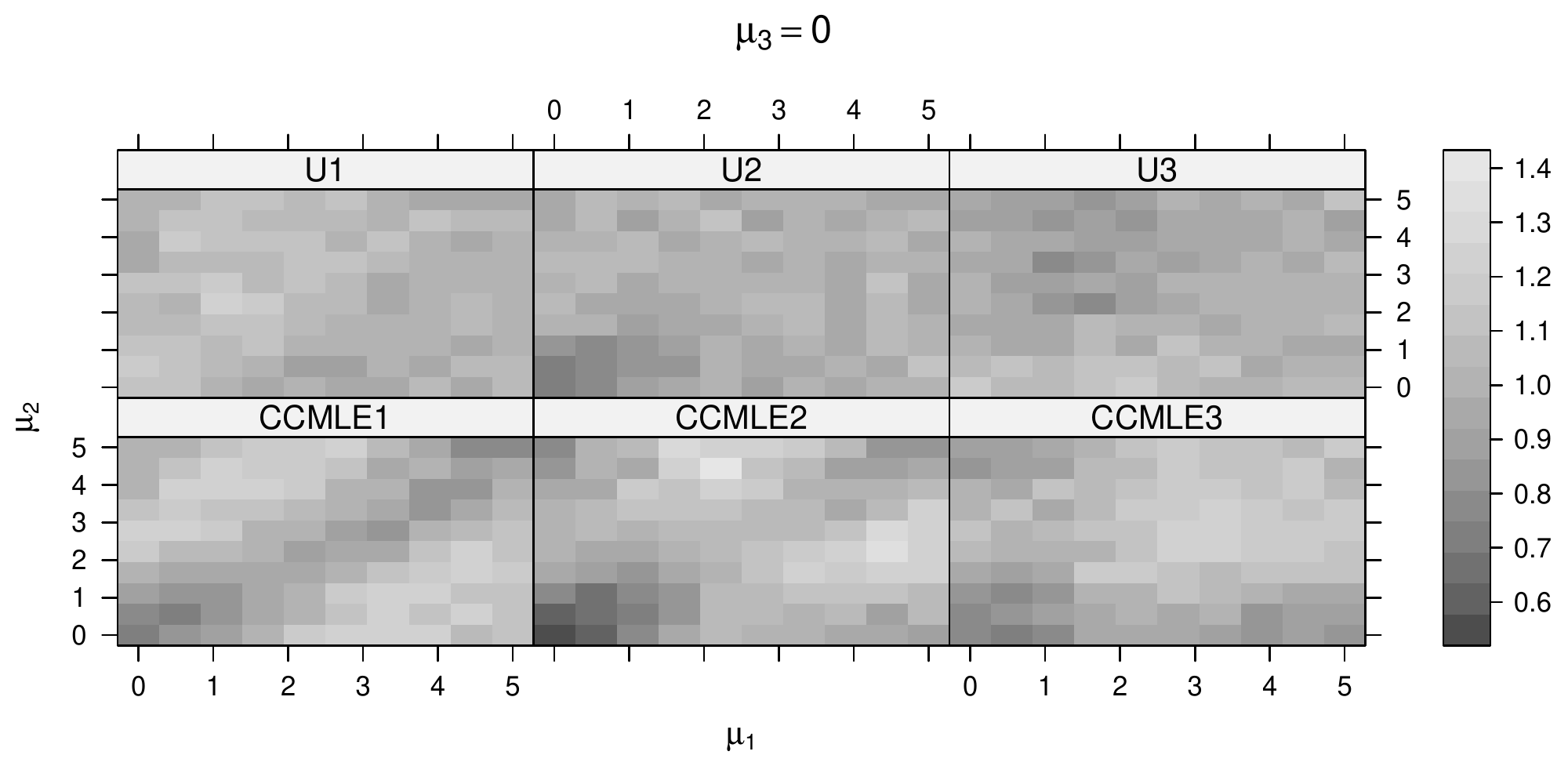}
\includegraphics[width=0.6\textwidth]{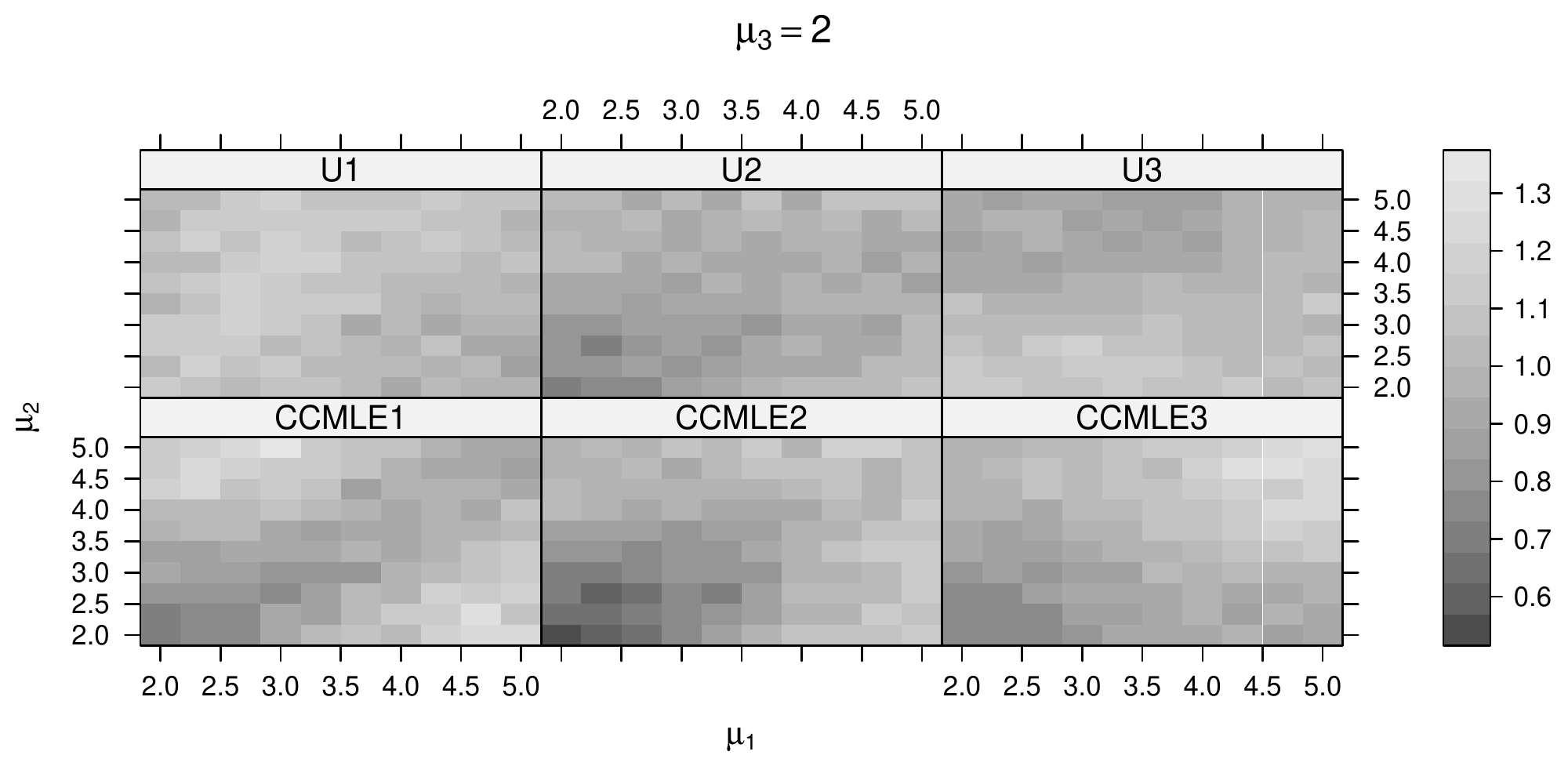}
\includegraphics[width=0.6\textwidth]{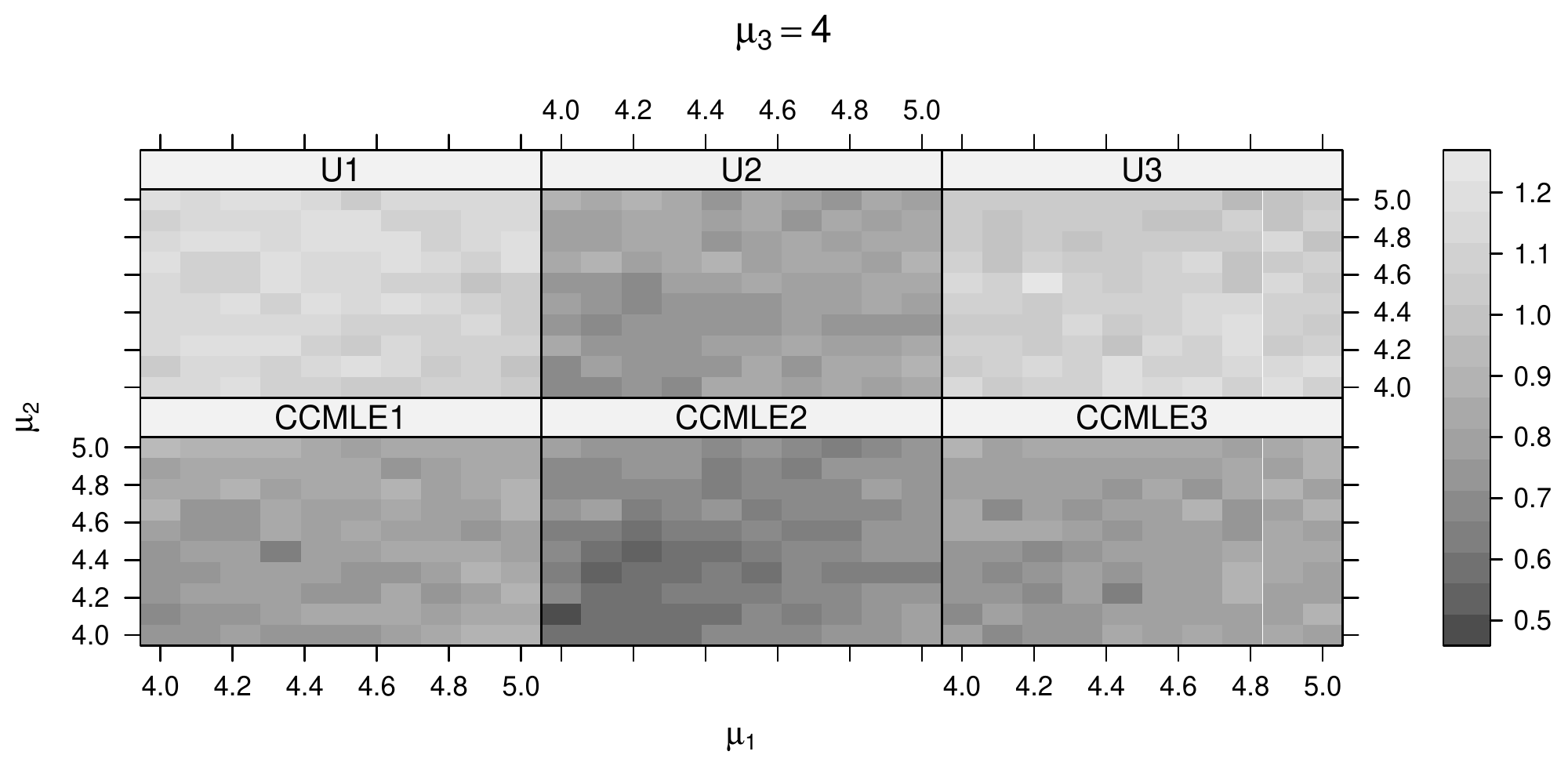}
\caption{Estimated MSE for $p=3$ simulation experiment}
\label{fig:plot_p3_1}
\end{figure}

Consider the three separate 2-by-3 lattice plots in Figure \ref{fig:plot_p3_1}. In
each, the level plots in the top row correspond to the ordinary MLE, and the
plots in the second row correspond to the CCMLE. The plots in the left-most
column correspond to the errors when estimating the mean of the population with
the \emph{maximum} sample mean, and the plots in the right-most column
correspond to the errors for the population with the \emph{minimum} sample mean.
Note that colors closer to black indicate a good estimator (low MSE) while
colors closer to white indicate a poorer performance (high MSE).

Let us focus first on the case $\mu_3=0$. Observe that there is a diagonal dark
line for the CCMLE in the left-most column. In the configurations on the
diagonal, the means of population 1 and 2 are close to each other, and hence the
CCMLE does well. In the off-diagonals, the populations are better separated and
hence performance goes down. Overall though, it does appear that the MLE and
CCMLE have a similar performance for this setting.

In the second configuration for $p=3$, we had fixed $\mu_3$ to be 2, while
$\mu_1$ and $\mu_2$ varied from 2 to 5. The level plots for this configuration
dark regions for the CCMLE than the MLE, when compared to the case when
$\mu_3=0$. The reason is that the number of configurations where at least two of
the populations are close together has increased, resulting in overall better
performance from the CCMLE.


Now focus on the displayed plots for $\mu_3=4$, in Figure \ref{fig:plot_p3_1}.
The difference between the CCMLE and the MLE is more pronounced when we consider
this final configuration, where $\mu_1$ and $\mu_2$ vary between 4 and 5. The
regions in the level plot for the CCMLE are consistently darker than those for
the MLE, due to the close proximity of sample means generated.

\subsection{Confidence Intervals based on the CCMLE}
\label{subsec:boot}

In this subsection, we use the stratified bootstrap to assess the confidence
intervals from the CCMLE procedure. Consider $p=3$, and the true means to be
$\mu_1$, $\mu_2$, $\mu_3$. We simulate a single sample of size 50 from each
population. We draw from $N(\mu_i, \sqrt{50})$ so that the sample mean still has
variance 1. Then we draw 9999 bootstrap samples from within each sample and
repeatedly compute the CCMLE. Each time, we are returned with an estimate of the
mean of the populations selected to be the maximum, middle and minimum. The bias
corrected intervals are then plotted for comparison with the traditional
intervals that do not incorporate the selection process. This procedure is
carried out for 4 different configurations of true means. The output can be seen
in Figure \ref{fig:boot_1}.

The output is most informative for the configuration where $\mu_1=10$, $\mu_2 =
9.5$ and $\mu_3 = 9$. In this case, the CCMLE point estimates are all equal,
since there is insufficient power to discriminate between them. The CCMLE intervals
are roughly the same, though it is worth pointing out the slight shrinkage, in
opposite directions, of the intervals for the maximum and the minimum. In the
situations where the mean of a group is far from the rest (for instance, in the
bars for ``Max'' in the plots on the right, and the bars for ``Min'' in the
bottom left), the CCMLE and the traditional approach provide comparable
intervals. In cases where means are not distinguishable, (for instance, in the
bars for ``Mid'' and ``Min'' in the top right), the CCMLE suggests it can provide 
shorter intervals. 

\begin{figure}[H]
\centering
\includegraphics[width=0.7\textwidth]{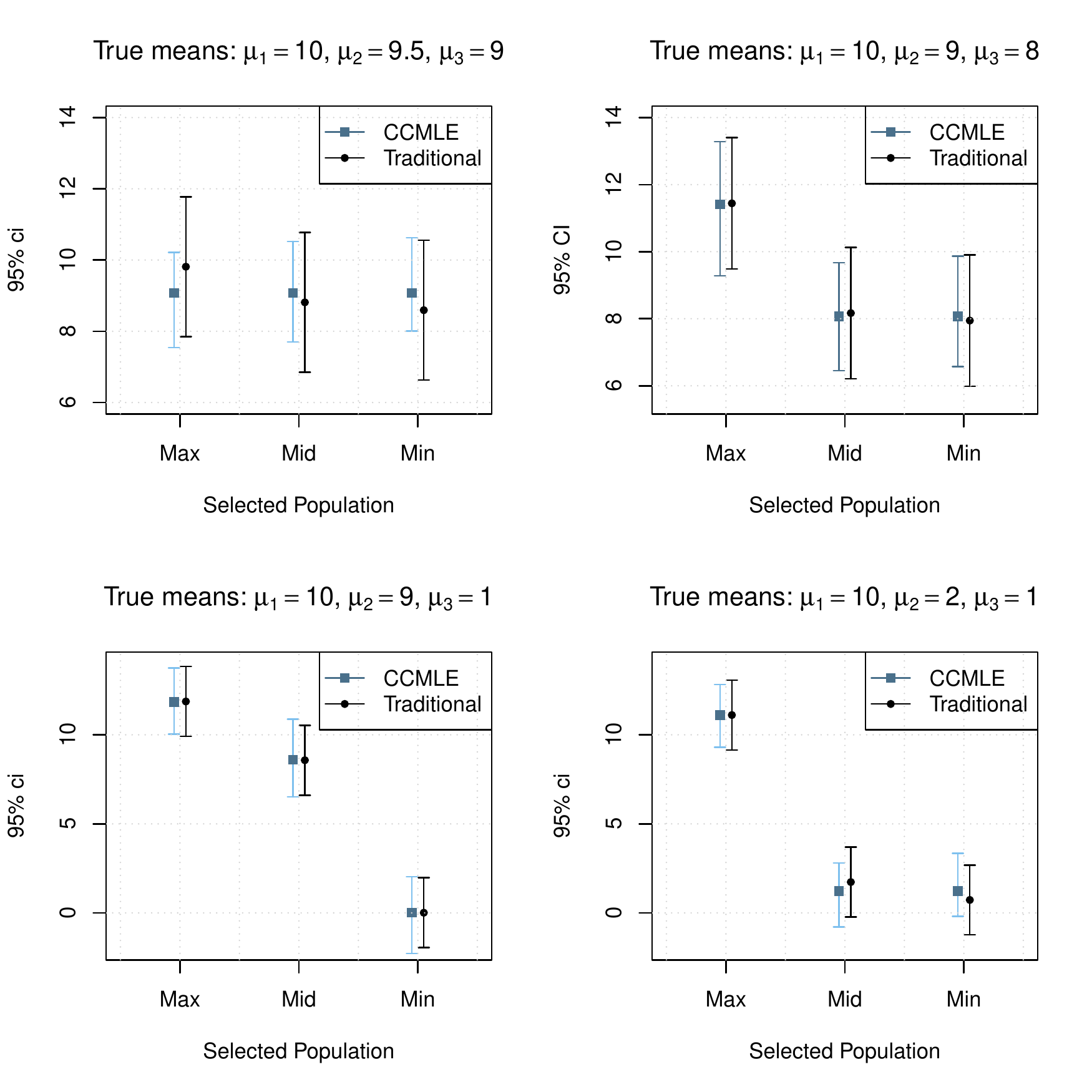}
\caption{Bootstrap confidence intervals using CCMLE procedure. See Section 
\ref{subsec:boot} for a detailed explanation.}
\label{fig:boot_1}
\end{figure}

\section{Discussion}

In this paper, we introduced a new estimator for the means of selected
populations and have started to unveil some of its properties. Although
simulations experiments suggest the estimator is not admissible (see Figures 2
and 3), it performs well, particularly when the population means are
close together. Furthermore, the proposed CCMLE provides simultaneous inference
on selected populations, without a need to pre-specify the number of populations
selected. Another advantage of the procedure is that because it is frequentist
in paradigm, there is no need for prior specification.

Although the focus in this paper has been on selection via ranking of sample
means, conceptually, this approach allows for any other selection criterion to
be used. For instance, if populations were to be selected based on the absolute
values of the sample means, the primary modification would be to the probability
in the denominator of equation (\ref{eq:cond_lik}). Finally, as presented, it is
straightforward to obtain bootstrap confidence intervals for this estimator, as
demonstrated in Section \ref{subsec:boot}.

Future work include the study of the asymptotic properties of this
estimator, and computationally efficient methods to approximate the probability
$P_{\bm{\mu}}(X_1>\ldots>X_p)$ in a general framework.

\bibliographystyle{plain}
\bibliography{ref}
\end{document}